\documentclass{article}
\usepackage[margin=2cm]{geometry}
\usepackage{amsmath,amssymb,amsfonts,amsthm,xcolor}
\usepackage{enumerate}

\newcounter{NN}
\newtheorem{proposition}[NN]{Proposition}
\newtheorem{theorem}[NN]{Theorem}

\newtheorem{lemma}[NN]{Lemma}

\newtheorem{example}[NN]{Example}

\def\b#1{{\bf #1}}
\def\x{{\bf x}}

\def\r{{\bf r}}
\def\k{{\bf k}}
\def\u{{\bf u}}
\def\A{{\bf A}}
\def\sp{{\rm Sp}}
\def\im{{\rm Im}}
\def\ker{{\rm Ker}}
\def\R{\mathbb{R}}
\def\e{{\rm e}}

\begin{document}

\title{On a quadratic Poisson algebra and integrable Lotka-Volterra systems with solutions in terms of Lambert's W function}
\author{
Peter H.~van der Kamp, D.I. McLaren and G.R.W. Quispel
\\[2mm]
Department of Mathematics and Statistics, La Trobe University, Victoria 3086, Australia.\\
Email: P.vanderKamp@LaTrobe.edu.au\\[7mm]
}

\maketitle

\begin{abstract}
We study a class of integrable inhomogeneous Lotka-Volterra systems whose quadratic terms are defined by an antisymmetric matrix and whose linear terms consist of three blocks. We provide the Poisson algebra of their Darboux polynomials, and prove a contraction theorem. We then use these results to classify the systems according to the number of functionally independent (and, for some, commuting) integrals. We also establish separability/solvability by quadratures, given the solutions to the 2- and 3-dimensional systems, which we provide in terms of the Lambert W function.  
\end{abstract}

\section{Introduction}
The quest for integrating ordinary differential systems with quadratic terms has a long history. In 1879, P. Hoyer studied the system
\begin{equation} \label{Hooi}
\begin{split}
\dot{x}_1 &= a_1x_2x_3 + a_2x_3x_1 + a_3x_1x_2 \\
\dot{x}_2 &= b_1x_2x_3 + b_2x_3x_1 + b_3x_1x_2 \\
\dot{x}_3 &= c_1x_2x_3 + c_2x_3x_1 + c_3x_1x_2,
\end{split}
\end{equation}
through Puiseux series expansion and found exact solutions in terms of elliptic functions \cite{Hoyer}. Kovalevskaya also studied the integrability of quadratic systems of ODEs before her seminal work on the Euler equations. In a letter to G. Mittag-Leffler \cite{Kov}, she discussed the class of  three-dimensional Lotka-Volterra (LV) systems
\begin{equation} \label{Koof}
\dot{x}_i=x_i\sum_{j=1}^3a_{ij}x_j,\quad i=1,2,3,
\end{equation}
mentioning that under the condition $a_{12}a_{23}a_{31}=a_{21}a_{32}a_{13}$ the general solution depends on 3 parameters and has only first order pole singularities in (finite) complex time, a property which is nowadays attributed to Painlev\'e, cf. \cite{Gor,LT,PS}. An $n$-component generalisation of the Kovalevskaya system (the case $n=3$, $k=2$, equation \eqref{Koof}),
\[
\dot{x}_i=x_i\left(-kx_i+\sum_{j=1}^nx_j\right),\quad i=1,\ldots,n,
\]
its integrability, and discretisation, was recently studied in \cite{BK,PS}. Generalising in another direction, a complete list of integrable inhomogeneous 3-component Lotka-Volterra systems was obtained in \cite{GMRSW}, using the Frobenius integrability theorem, Painlev\'e analysis, and the Jacobi last multiplier method. General solutions of many of these systems, including certain
ABC and May–Leonard systems, were constructed with the aid of special functions, in \cite{Mai}.

Another (not unrelated) approach to integrability is via the theory of Darboux polynomials. This is the approach taken in e.g. \cite{Cai,ChD} and the approach we adopt here. A polynomial $P(x)$ is a Darboux polynomial (DP) if $\dot{P}(x)=P(x)C(x)$, where $C(x)$ is called the cofactor of $P(x)$. A DP with cofactor $0$ is an integral, and the numerator and denominator of a rational integral are DPs with the same cofactor. DPs were already studied by Darboux, Poincaré, Painlevé and others, and are known by several other names, including ‘second integrals’ and ‘weak integrals’, cf. \cite{GorB} for a nice introduction.

In \cite{KKQTV}, the Lotka-Volterra system
\begin{equation} \label{OLV}
\dot{x}_i=x_i\,\big(\sum_{j>i} x_j - \sum_{j<i} x_j\big), \qquad i=1,\ldots,n,
\end{equation}
arose as a subsystem of quadratic vector fields associated with multi-sums of products, and it was shown (for all $n$) to be superintegrable (with $n-1$ integrals) as well as Liouville integrable.
This generalised the observation of Bogoyavlenski\v{\i}, that the $n=4$ case of \eqref{OLV} is superintegrable \cite[Equation (7.8)]{Bog}. Indeed, the system \eqref{OLV} can be derived as a reduction of the following ODE on a 1-dimensional lattice \cite{Bog}
\begin{equation} \label{bog}
\dot{x}_i=x_i\,\big(\sum_{j=1}^{n-1} x_{i+j} - \sum_{j=1}^{n-1} x_{i-j}\big),
\end{equation}
which generalises the well-known Volterra lattice ($n=2$ case, discrete KdV equation)
$\dot{x}_i=x_i\,\big(x_{i+1} - x_{i-1}\big)$.
Namely, by considering the periodic $x_{i+2n-1}=x_i$ lattice \eqref{bog} on the invariant submanifold $x_{n+1}=x_{n+2}=\cdots=x_{2n-1}=0$ one obtains the system \eqref{OLV}. We also note that the system \eqref{OLV} possesses the $n \choose 2$ DPs $\sum_{k=i}^j x_k$, $1\leq i<j\leq n$. Using the fact that a homogeneous LV system in $n$ dimensions only needs $n-1$ additional DPs to be superintegrable, the system \eqref{OLV} was generalised in \cite{treeL,treeP} to families of $(3n-2)$-parameter superintegrable systems, in one-to-one correspondence with trees on $n$ vertices.

In \cite{KMQ} it was shown that systems of ODEs can also be modified by adding a term $b(\x; t)\x$, where $b$ is a scalar function of $\x,t$, whilst preserving its homogeneous Darboux polynomials. And, using that result, the inhomogeneous generalisation of the system (\ref{OLV}), in even dimensions,
\begin{equation} \label{GLV}
\dot{x}_i=x_i\,\big(b + \sum_{j>i} x_j - \sum_{j<i}x_j\big), \qquad i=1,\ldots,2m=n,
\end{equation}
was shown to be superintegrable as well as Liouville integrable. Kouloukas, Quispel and Vanhaecke \cite{KQV} proved another generalisation of (\ref{OLV}),
\begin{equation} \label{KQV}
\dot{x}_i=x_i\,\big(\sum_{j>i} a_jx_j - \sum_{j<i} a_jx_j\big), \qquad i=1,\ldots,n,
\end{equation}
to be Liouville and superintegrable, and Christodoulidi, Hone, Kouloukas \cite{CHK} considered systems of the
form
\begin{equation} \label{CHK}
\dot{x}_i=x_i\,\big(r_i+\sum_{j>i} a_jx_j - \sum_{j<i} a_jx_j \big), \qquad i=1,\ldots,n.
\end{equation}
They derived equations for $a_j$ and $r_i$ such that (\ref{CHK}) admits the same integrals as (\ref{KQV}).

In this paper we study inhomogeneous LV systems of the form
\begin{equation} \label{LVF}
\dot{x}_i=x_i\,\big(r_i + \sum_{j>i} x_j - \sum_{j<i}x_j\big),\qquad
r_i=\begin{cases}
b, & i=1,\ldots,k,\\
c, & i=k+1,\ldots, k+l,\\
d, & i=k+l+1,\ldots, k+l+m=n,
\end{cases}
\end{equation}
where $c=b+d$ and $km\neq 0$. We will say that such a system has three blocks (values of the parameters $b,c,d$), or two blocks, if $l=0$. We will adopt the following definitions. An $n$-dimensional system of ODEs is 
\begin{itemize}
\item {\em Liouville integrable:} there are $r/2+s$ functionally independent integrals which Poisson commute, of which $s=n-r$ are Casimir functions and $r$ is the rank of the Poisson bracket.
\item {\em Nonholonomically integrable:} the system is measure-preserving and admits $n-2$ functionally independent integrals, cf. \cite[Theorem 13]{Arnold}.
\item {\em Liouville superintegrable:} the system is Liouville integrable, and admits more functionally independent integrals than required.
\item {\em Liouville maximally superintegrable:} the system is Liouville integrable, and admits $n-1$ functionally independent integrals.
\item {\em Integrable-or-not:} the system is not a priori Liouville integrable, but admits more than $r/2+s$ functionally independent integrals.
\end{itemize} 
Each notion of integrability implies, in principle, that the system can be integrated by quadratures. For systems which are integrable-or-not, this is not necessarily the case.

We have identified nine integrable families of the form \eqref{LVF}. In even dimensions, we define
the families E1-E5:
\begin{align}
l=0, k\equiv m \equiv 1 \tag{E1}\label{E1}\\
l=0, k\equiv m \equiv 0 \tag{E2}\label{E2}\\
k\equiv 0, l\equiv m \equiv 1 \tag{E3}\label{E3}\\
k\equiv m \equiv 1, l\equiv 0 \tag{E4}\label{E4}\\
k\equiv l \equiv m \equiv 0 \tag{E5}\label{E5},
\end{align}
in odd dimensions we define O1-O3:
\begin{align}
k\equiv l\equiv m \equiv 1 \tag{O1}\label{O1}\\
k\equiv 1, l\equiv m\equiv 0, b=0 \tag{O2}\label{O2}\\
k\equiv m\equiv 0, l\equiv 1, c=0 \tag{O3}\label{O3},
\end{align}
and N1 is given by
\begin{align}
l=0, k\equiv 1, m\equiv 0 \tag{N1}\label{N1},
\end{align}
where equivalence is taken modulo 2. We remark that the case $m\equiv 0, k\equiv l \equiv 1$ is equivalent to \ref{E3}, the case $k\equiv l\equiv 0, m\equiv 1, d=0$ is equivalent to the case \ref{O2}, and the case $l=0, k\equiv 0, m\equiv 1$ is equivalent to the case \ref{N1}, by a linear transformation.

We show that the E (even) and O (odd) cases are Liouville integrable, and that the N case is Nonholonomically integrable. We further classify the O and E systems with respect to their number of integrals; some are Liouville maximally superintegrable, some are nonholonomically integrable and others only have $n-3$ integrals. We also define the following cases:
\begin{align}
k\equiv 1, l\equiv m\equiv 0, b\neq 0  \tag{M1}\label{M1}\\
k\equiv m\equiv 0, l\equiv 1, c\neq 0 \tag{M2}\label{M2},
\end{align}
and note that the case $k\equiv l\equiv 0, m\equiv 1, d\neq 0$ is equivalent to \ref{M1}. Together the O,E,N, and M cases comprise the family of systems \eqref{LVF}. We also show that the entire family is separable and solvable by quadratures, given the solution to the special cases $k=l+1=m=1$ and $k=l=m=1$. This is somewhat surprising as it seems that the M systems are neither Liouville integrable, nor nonholonomically integrable.

The paper is organised as follows.
In section \ref{SecPA}, we define two sets of Darboux functions, sums of consecutive variables and products of consecutive variables with alternating exponents, and provide the quadratic Poisson algebra they satisfy (with respect to the bracket given in \cite[Equation (3.6)]{KKQTV}. We identify functions (on even-dimensional spaces)  that act as scalings.

In section \ref{SecH}, we provide the Hamiltonian theory of the general class of systems of inhomogeneous LV systems, and show that for systems whose quadratic part is the same as \eqref{LVF} the Hamiltonian (if it exists) can be written in terms of the scaling functions.

In section \ref{SecA}, we provide an alternative proof for the Liouville maximally superintegrability of equation \eqref{GLV} and of a similar but odd-dimensional system. We distinguish between integrals that depend on all the variables (which could be called global), and integrals that depend on a subset of the variables within a block, i.e. variables $x_i,x_{i+1},\ldots,x_j$, with $r_i=r_{i+1}=\cdots=r_{j}$. These will be called local integrals. Using local integrals the proofs of involutivity and independence simplify quite drastically.
In section \ref{SecC} we explain how for given values of $k,l$ and $m$, one can find out how many independent integrals exist, that are products of powers of the Darboux functions and an exponential factor. In higher dimensional systems, we find that the bulk of the independent integrals can be given in terms of local integrals. However, in all cases one needs a few global ones (at least one, at most three). The origin of these global integrals can be understood by what we have dubbed the contraction theorem. By a linear transformation, we can reduce the system to a system where one or more blocks are smaller in size. Integrals from the contracted system can be lifted to integrals for the original system. And, conveniently, lifted integrals commute with local integrals.

In section \ref{SecI}, we provide sets of sufficiently many integrals for the various notions of integrability. These sets are inferred from integrals we have calculated in low dimensions. In section \ref{SecP}, we prove that the sets of integrals postulated in section \ref{SecI}, are functionally independent in any dimension, and, in the Liouville integrable cases, sufficiently many of them Poisson commute. In section \ref{SecS}, we provide the general solution for the two and three dimensional systems of the form \eqref{LVF} with $(k,l,m)=(1,0,1)$ and $(k,l,m)=(1,1,1)$, and we provide a separation of variables for the general system in any dimension. The final section comprises a summary and outlook.
\section{A quadratic Poisson algebra} \label{SecPA}
All known integrals of the systems (\ref{OLV}-\ref{GLV}) can be expressed in terms of
\begin{equation} \label{PsQs}
P_{i,j}=x_i+x_{i+1}+\cdots+x_{j},\qquad
Q_{i,j}=x_i^{-1}x_{i+1}x_{i+2}^{-1}\cdots x_{j}^{(-1)^{j-i+1}},
\end{equation}
where it is understood that $P_{i,j}=0$ and $Q_{i,j}=1$ when $i>j$.
The Poisson brackets of these functions, with respect to
\begin{equation} \label{bra}
\{x_i,x_j\}=x_i\,x_j,\qquad i<j,
\end{equation}
are provided in the following lemma, which is proved in Appendix A.
\begin{lemma} \label{PA}
Let $[m,p]=[i,j]\cap [k,l]$, and let, taking $k\leq l$,
\begin{equation} \label{ckl}
c^{k,l}_a=\begin{cases} 0 & a<k \\ 1 & k\leq a \leq l \\ 0 & a > l, \end{cases}
\text{ when } l\not\equiv k, \text{ and }
c^{k,l}_a=\begin{cases} -1 & a<k \\ 0 & k\leq a \leq l \\ 1 & a > l \end{cases}
\text{ when } l\equiv k,
\end{equation}
where, here and in the sequel, equivalence is taken modulo 2.
The Poisson algebra of the functions (\ref{PsQs}) is a quadratic Poisson algebra:
\begin{align}
\{P_{i,j},P_{k,l}\}=&P_{i,j}\,P_{k,l}-P_{m,p}^2-2\,P_{m,p}\,P_{l+1,j}\qquad (k\geq i) \label{PP}\\
\{P_{i,j},Q_{k,l}\}=&\left(\sum_{a=i}^j c^{k,l}_a\, x_a \right) Q_{k,l} \label{PQ}\\
\{Q_{i,j},Q_{k,l}\}=&\left(\sum_{a=i}^j (-1)^{a+i+1}\, c^{k,l}_a \right) Q_{i,j}\,Q_{k,l}. \label{QQ}
\end{align}
\end{lemma}

\noindent
The sum in (\ref{PQ}) can be expressed as
\[
\sum_{a=i}^j c^{k,l}_a\, x_a=
\begin{cases}
P_{m,p} & k \not\equiv l,\\
-P_{i,\min(j,k-1)}+P_{\max(l+1,i),j} & k\equiv l.
\end{cases}
\]
The sum in (\ref{QQ}) has values $0$ or $\pm1$. We have
\[
\sum_{a=i}^j (-1)^{a+i+1} c^{k,l}_a = \begin{cases} (-1)^{m+i+1}\,\#_{m,p} & k \not\equiv l, \\
\#_{i,\min(j,k-1)}-(-1)^{\max(i,l+1)+i}\#_{\max(i,l+1),j} & k\equiv l,
\end{cases}
\]
where $\#_{m,p}\in\{0,1\}$ tells whether the interval $[m,p]$ contains an even or an odd number of integers. Taking $i=j$, $k=1$, and $l=n$ odd, we find $\{x_i,Q_{1,n}\}=0$ for all $1\leq i\leq n$ which shows that $Q_{1,n}$ is a Casimir function.

A {\em scaling}, with respect to variable $x_j$, on a Poisson manifold is a function $\sigma_j$ such that when $P$ is a homogeneous polynomial in $x_j$, of weight (or degree) $d$, the Poisson bracket with $\sigma_j$ is $\{P,\sigma_j\}=dP$.

\begin{lemma} \label{SC}
On the even $n$-dimensional Poisson manifold with bracket (\ref{bra}), the functions
\begin{equation} \label{sigma}
\sigma_j=\ln\left(Q_{1,j-1}^{(-1)^j}\,Q_{j+1,n}^{-1}\right)
\end{equation}
are scalings. The total scaling is
\begin{equation} \label{equality}
\sigma = \sum_{j=1}^{n} \sigma_j = \ln\left(Q_{1,n}\right).
\end{equation}
\end{lemma}
\begin{proof}
We have $\{x_k,\sigma_j\}=x_k\delta_{k,j}$, which follows from Lemma \ref{PA}. The equality (\ref{equality}) is a special case of
\begin{equation} \label{PS}
\sum_{i=k}^l \sigma_i=\begin{cases}
\ln\left(Q_{k,l}\right) & k\not\equiv l \\
\ln\left(Q_{1,k-1}^{(-1)^k}\,Q_{l+1,n}^{-1}\right)& k\equiv l, \\
\end{cases}
\end{equation}
which is proved by induction.
\end{proof}
\noindent
The expression (\ref{PS}) will be referred to as a partial scaling, and denoted by $\sigma_{k}^{l}=\sum_{i=k}^l \sigma_i$.

\section{Hamiltonians for LV-systems} \label{SecH}
The family of Lotka-Volterra equations
\begin{equation}\label{LVgen}
\dot{x}_i = x_i(r_i +  \sum_{j=1}^{n} a_{ij} x_j)
\end{equation}
with $a_{ij}$ being the elements of a matrix $\A$, is Hamiltonian, or Poisson, if $\A$ is skew and $\r\in\im(\A)$ \cite[Section 3.14]{MQ}. In terms of the Poisson bracket $\{x_i,x_j\}=a_{ij}x_ix_j$, the system (\ref{LVgen}) can be written as $\dot{\x}=\{\x,H\}$, where the Hamiltonian is given by
\begin{equation}\label{hamLV}
H = \sum_{i=1}^{n} (x_i + k_i \ln x_i), \qquad \text{ with } \A\k=\r.
\end{equation}
For equations of the form
\begin{equation} \label{FGLV}
\dot{x}_i=x_i\,\big(r_i + \sum_{j>i} x_j - \sum_{j<i}x_j\big),
\end{equation}
the $n\times n$ skew matrix $\A$ is given by
\begin{equation}\label{AM}
\A=\begin{pmatrix}
0 & 1 & 1 & \cdots & 1 \\
-1 & 0 & 1 & \cdots & 1 \\
-1 & -1 & 0 & \cdots & 1 \\
\vdots & \vdots & \vdots & \ddots & \vdots \\
-1 & -1 & -1 & \cdots &0
\end{pmatrix}
\end{equation}
and the above Poisson bracket coincides with \eqref{bra}. We distinguish two cases.
\begin{itemize}
\item The Hamiltonian case. When $n$ is even the matrix $\A$ has full rank and its inverse is
\[
\A^{-1}=\begin{pmatrix}
0 & -1 & 1 & \cdots & -1 \\
1 & 0 & -1 & \cdots & 1 \\
-1 & 1 & 0 & \cdots & -1 \\
\vdots & \vdots & \vdots & \ddots & \vdots \\
1 & -1 & 1 & \cdots &0
\end{pmatrix},
\]
so that, for $i=1,\ldots,n$,
\begin{equation} \label{ki}
k_i=-\sum_{j=1}^{i-1}(-1)^{i+j}r_j + \sum_{j=i+1}^{n}(-1)^{i+j} r_j.
\end{equation}
\item The Poisson case. When $n$ is odd the rank of $\A$ is $n-1$, and with $\u=(1,-1,1,\ldots,-1,1)^T$ we have $\ker(\A)=\sp(\u)$. The vector $\r$ is in the image of $\A$ if
\begin{equation} \label{rc}
r_n=\sum_{i=1}^{n-1}(-1)^i r_i,
\end{equation}
in which case we obtain, for $i=1,\ldots,n-1$,
\begin{equation} \label{kio}
k_i=u_i k_n - \sum_{j=1}^{i-1}(-1)^{i+j}r_j + \sum_{j=i+1}^{n-1}(-1)^{i+j} r_j,
\end{equation}
and $k_n$ is a free parameter.
\end{itemize}
The Hamiltonian of the Lotka-Volterra equation \eqref{FGLV} can be conveniently expressed in terms of the
parameters of the system, and scalings.
\begin{proposition} \label{HP}
When $n$ is even, we have
\[
H=\sum_{j=1}^{n} \left( x_j + r_j \sigma_j \right ),
\]
and, when $n$ is odd and $\r$ satisfies (\ref{rc}), we have
\[
H=\sum_{j=1}^{n} x_j + \sum_{j=1}^{n-1} r_j \sigma_j + f(Q_{1,2m+1}),
\]
where $f$ is an arbitrary function. 
\end{proposition}
\begin{proof}
Let $n$ be even and consider the sum of terms linear in $\k$ in (\ref{hamLV}). Using \eqref{ki}, we find
\begin{align*}
\sum_{i=1}^n k_i \ln x_i &= \ln \left( \prod_{i=1}^n x_i^{k_i} \right) \\
&= \ln \left( \prod_{i=1}^n x_i^{\sum_{j=1}^{i-1}(-1)^{i+j+1}r_j + \sum_{j=i+1}^{n}(-1)^{i+j} r_j} \right)\\
&= \ln \left( \prod_{i=1}^n \left( \prod_{j=1}^{i-1} x_i^{(-1)^{i+j+1}r_j}\prod_{j=i+1}^{n}x_i^{(-1)^{i+j} r_j} \right)\right)\\
&= \ln \left( \prod_{j=1}^{n-1} \prod_{i=j+1}^{n} x_i^{(-1)^{i+j+1}r_j} \right) + \ln \left( \prod_{j=2}^{n} \prod_{i=1}^{j-1} x_i^{(-1)^{i+j} r_j} \right)\\
&= \sum_{j=1}^{n-1} r_j \ln \left( \prod_{i=j+1}^{n} x_i^{(-1)^{i+j+1}} \right) + \sum_{j=2}^{n} r_j \ln \left( \prod_{i=1}^{j-1} x_i^{(-1)^{i+j}} \right)\\
&= \sum_{j=1}^{n} r_j \ln \left( Q_{j+1,n}^{-1} Q_{1,j-1}^{(-1)^j} \right).
\end{align*}
For the odd case, from \eqref{kio}, we get an extra multiple of the Casimir $1/Q_{1,n}$ which we changed into an arbitrary function of $Q_{1,n}$.
\end{proof}
Proposition \ref{HP}, together with equation \eqref{PS}, enables us to write Hamiltonians for Lotka-Volterra systems of the form \eqref{LVF} in terms of the functions of $P$ and $Q$, \eqref{PsQs}. For \eqref{OLV} we have $H_{\eqref{OLV}}=P_{1,n}$. For \eqref{GLV} we obtain $H_{\eqref{GLV}}=P_{1,n}+b\sigma=P_{1,n}+b\ln(Q_{1,n})$. For the odd $n$ variant, defined by
\begin{equation} \label{OGLV}
\dot{x}_i=x_i\,\big(r_i + \sum_{j>i} x_j - \sum_{j<i}x_j\big),\qquad
r_i=\begin{cases}
b, & i=1,\ldots,n-1,\\
0, & i=n,
\end{cases}
\end{equation}
we find $H_{\eqref{OGLV}}=P_{1,n}+b\sigma_1^{n-1}=P_{1,n}+b\ln(Q_{1,n-1})$. This system, \eqref{OGLV}, is equivalent to \ref{O2} with $k=1$. In the next section, we consider the systems \eqref{GLV} and \eqref{OGLV} as a warm-up exercise. For the E and O cases we obtain
\begin{align*}
H_{\text{\ref{E1}}}
    &=P_{1,n}+b\sigma_{1}^{k}+d\sigma_{k+1}^{n}\\
    &=P_{1,n}-b\ln(Q_{k+1,n})+d\ln(Q_{1,k}),\\[1mm]
 H_{\text{\ref{E2}}}
    &=P_{1,n}+b\sigma_{1}^{k}+d\sigma_{k+1}^{n}\\
    &=P_{1,n}+b\ln(Q_{1,k})+d\ln(Q_{k+1,n}),\\[1mm]
H_{\text{\ref{E3}}}
    &=P_{1,n}+b\sigma_{1}^{k}+(b+d)\sigma_{k+1}^{k+l}+d\sigma_{k+l+1}^{n}\\
    &=P_{1,n}+b\ln(Q_{1,k})-(b+d)\ln(Q_{1,k}Q_{k+l+1,n})+d\ln(Q_{1,k+l})\\
    &=P_{1,n}-(b+d)\ln(Q_{k+l+1,n})+d\ln(Q_{k+1,k+l}),\\[1mm]
 H_{\text{\ref{E4}}}
    &=P_{1,n}+b\sigma_{1}^{k}+(b+d)\sigma_{k+1}^{k+l}+d\sigma_{k+l+1}^{n}\\
    &=P_{1,n}-b\ln(Q_{k+1,n})+(b+d)\ln(Q_{k+1,k+l})+d\ln(Q_{1,k+l})\\
    &=P_{1,n}-b\ln(Q_{k+l+1,n})+d\ln(Q_{1,k}),\\[1mm]
 H_{\text{\ref{E5}}}
    &=P_{1,n}+b\sigma_{1}^{k}+(b+d)\sigma_{k+1}^{k+l}+d\sigma_{k+l+1}^{n}\\
    &=P_{1,n}+b\ln(Q_{1,k})+(b+d)\ln(Q_{k+1,k+l})+d\ln(Q_{k+l+1,n},\\[1mm]
 H_{\text{\ref{O1}}}
    &=P_{1,n}+b\sigma_{1}^{k}+(b+d)\sigma_{k+1}^{k+l}+d\sigma_{k+l+1}^{n-1}\\
    &=P_{1,n}-b\ln(Q_{k+1,n-1})+(b+d)\ln(Q_{1,k}/Q_{k+l+1,n-1})+d\ln(Q_{k+l+1,n-1})\\
    &=P_{1,n}+b\ln(Q_{1,k+l})+d\ln(Q_{1,k}),\\[1mm]
H_{\text{\ref{O2}}}
    &=P_{1,n}+d\sigma_{k+1}^{k+l}+d\sigma_{k+l+1}^{n-1}\\
    &=P_{1}^{n}+d\ln(Q_{1}^{k}/Q_{k+l+1}^{n-1})+d\ln(Q_{k+l+1}^{n-1})\\
    &=P_{1}^{n}+d\ln(Q_{1}^{k}),\\[1mm]
 H_{\text{\ref{O3}}}
    &=P_{1}^{n}+b\sigma_{1}^{k}-b\sigma_{k+l+1}^{n-1}\\
    &=P_{1,n}-b\ln(Q_{k+1,n-1})-b\ln(Q_{k+l+1,n-1})\\
    &=P_{1,n}-b\ln(Q_{k+1,k+l}).
\end{align*}
The systems \ref{N1}, \ref{M1}, \ref{M2} are not Hamiltonian.

\section{Alternative proof for the integrability of (\ref{GLV}), and its odd variant} \label{SecA}
The functions (\ref{PsQs}) are Darboux functions for (\ref{OLV}). We obtain their cofactors, with $1\leq i\leq j \leq n$, from
\begin{equation} \label{CP}
\dot{P}_{i,j}=\{P_{i,j},H_{(\ref{OLV})}\}=(-P_{1,i-1}+P_{j+1,n})\,P_{i,j}
\end{equation}
and
\begin{equation} \label{CQ}
\dot{Q}_{i,j}=\{Q_{i,j},H_{(\ref{OLV})}\}=-\left(\sum_{a=1}^n c^{i,j}_a\, x_a \right)Q_{i,j}.
\end{equation}
The integrals that were used in \cite{KKQTV}, to prove the integrability of \eqref{OLV}, are given in terms of (\ref{PsQs}) as
\[
F_k=P_{1,2k}\,Q_{2k+1,n},\qquad G_k=P_{n-2k+1,n}/Q_{1,n-2k}.
\]
These integrals are homogeneous (with weight 1) and therefore\footnote{\cite[Thm 3.1]{KMQ}: If $P(x)$ is a homogeneous Darboux polynomial of degree $d$ with cofactor $C(x)$ for the system of ODEs $\dot{x} = f (x)$, then $P(x)$ is a Darboux polynomial for the system $\dot{x} = f (x)+b(x, t)x$, with cofactor $C(x)+db(x; t)$, where $b$ is a scalar function of $x$.} they are Darboux functions, with cofactor $b$, of system (\ref{GLV}). In \cite{KMQ}, the integrability of (\ref{GLV}) was proved using integrals $F^\prime_k=F_k/H_{(\ref{OLV})}$, $G^\prime_k=G_k/H_{(\ref{OLV})}$.

We provide an alternative proof of \cite[Theorem 4.1]{KMQ}, which states that the even dimensional system (\ref{GLV}) is both superintegrable and Liouville integrable, and prove a similar statement for the related odd dimensional system (\ref{OGLV}). This will illustrate the approach we take in this paper.

When $i \equiv j$ the cofactors of $P_{i,j}$ and $Q_{i,j}$, with respect to (\ref{OLV}), add up to zero. Hence, each product
\begin{equation} \label{Iij}
I_{i,j}=P_{i,j}\,Q_{i,j},\text{ with } i \equiv j,
\end{equation}
is an integral for (\ref{OLV}). These integrals are local, in the sense that they do not depend on all variables. Consequently they are integrals for infinitely many vector-fields (for any $n\geq j$). 
Moreover, they are homogeneous of weight 0, and so they are also integrals of the system (\ref{GLV}). We will employ these integrals as much as possible. 

\begin{theorem} \label{FT}
Both the even dimensional system (\ref{GLV}), and the odd dimensional system (\ref{OGLV}) are Liouville integrable and maximally superintegrable.
\end{theorem}

\begin{proof} We first deal with the even dimensional system.
\begin{itemize}
\item[$n\equiv 0$] {\em Liouville integrability.} We show that each of the sets
\begin{equation} \label{sets}
\{H_{(\ref{GLV})}\} \cup \{I_{1,1+2i}\}_{i=1}^{n/2-1},\quad \{H_{(\ref{GLV})}\} \cup \{I_{2,2+2i}\}_{i=1}^{n/2-1}
\end{equation}
is a functionally independent set of $n/2$ pairwise commuting integrals.  Consider the bracket
\begin{equation} \label{br}
\{I_{i,j},I_{k,l}\}=
\{P_{i,j},P_{k,l}\}Q_{i,j}Q_{k,l}+
\{P_{i,j},Q_{k,l}\}Q_{i,j}P_{k,l}+
\{Q_{i,j},P_{k,l}\}P_{i,j}Q_{k,l}+
\{Q_{i,j},Q_{k,l}\}P_{i,j}P_{k,l}
\end{equation}
When $i\equiv j\equiv k\equiv l \mod 2$, and $[k,l]$ is a subset of $[i,j]$, i.e. $[i,j]\cap[k,l]=[k,l]$, it follows from Lemma \ref{PA} that the first term equals minus the second term, and that the last two terms both vanish. Functional independence follows from the fact that the integrals can be ordered so that each integral depends on one or more new variables, e.g., for the first set, $I_{1,3}$ depends on $x_1,x_2,x_3$, $I_{1,5}$ depends also on $x_4,x_5$, $\ldots$, $I_{1,n-1}$ depends also on $x_{n-2},x_{n-1}$, and finally $H_{(\ref{GLV})}$ depends also on $x_{n}$.

{\em Superintegrability.} To show that the union of the sets (\ref{sets}) is a functionally independent set of $n-1$ integrals we note that each integral (\ref{Iij}) can be written in terms of homogeneous variables
$y_i=x_i/x_{i+1}$, $i=1,\ldots,n-1$. In terms of variables $y_1,\ldots,y_{n-1},x_n$ the integral $I_{1,3}$ depends on $y_1,y_2$, $I_{2,4}$ depends also on $y_3$, $\ldots$ , $I_{n-2,n}$ depends also on $y_{n-1}$ and $H_{(\ref{GLV})}$ depends also on $x_{n}$.
\item[$n\equiv 1$] {\em Liouville integrability.} 
The sets
\begin{equation} \label{osets}
\{H_{\eqref{OGLV}},Q_{1,n}\} \cup \{I_{1,1+2i}\}_{i=1}^{n/2-1},\quad \{H_{\eqref{OGLV}},Q_{1,n}\} \cup \{I_{2,2+2i}\}_{i=1}^{n/2-1}
\end{equation}
are functionally independent sets of $(n-1)/2+1$ pairwise commuting integrals (including 1 Casimir).

{\em Superintegrability.} Lastly, we prove that the union of the sets (\ref{osets}) is functionally independent. In terms of variables $y_1,\ldots,y_{n-1},x_n$, the $n-3$ local integrals depend only on $y_1,\ldots,y_{n-2}$. Hence, it suffices to show that the Jacobian of $H_{\eqref{OGLV}},Q_{1,n}$ with respect to $y_{n-1},x_n$ has full rank. Using the inverse transformation,
$x_i=x_n\prod_{j=i}^{n-1} y_j$, $i<n$, we obtain
\[
H_{\eqref{OGLV}}=x_n \left( \sum_{i=1}^{n} \prod_{j=i}^{n-1} y_j \right) + b\ln\left(\frac{1}{\prod_{j=1}^{(n-1)/2} y_{2i-1}}\right),\quad
Q_{1,n}=\frac{1}{x_n\prod_{j=1}^{(n-1)/2} y_{2i-1}}.
\]
At $y_i=x_n=1$, we find
\[
\frac{D(H_{\eqref{OGLV}},Q_{1,n})}{D(y_{n-1},x_n)}=\begin{pmatrix} n-1 & n \\ 0 & -1 \end{pmatrix},
\]
which has full rank for $n\neq 1$.
\end{itemize}
\end{proof}

\section{An example with 3 blocks, and contraction} \label{SecC}
One can show that the polynomial $P=\sum_{i=1}^n c_ix_i$ is a linear Darboux polynomial of LV-system \eqref{FGLV} if and only if there exists $k,l: r_i=r_j, c_i=c_j$ if $k\leq i,j\leq l$ and $c_m=0$ if $m<k$ or $m>l$. And, the expression $\exp(P)$ is an exponential factor with linear cofactor if and only if $c_i=c_j$ for all $1\leq i,j\leq n$. 

This enables one to find out, in each case, how many integrals exist, of the form $\prod_i P_i^{\alpha_i}$ with linear Darboux polynomials or exponential factor $P_i$. We write down the set of linear DPs and exponential factors, and determine their cofactors, which are affine functions. We then create the matrix of coefficients of the cofactors, $M$, and determine a basis for the kernel of $M^T$. Each basis vector gives rise to one integral, and one can check how many of them are functionally independent.

\begin{example}
Consider the system \eqref{LVF} with $k=2,l=3,m=1$. There are 10 Darboux polynomials,
\[ 
P_1=x_1, P_2=x_2, \ldots ,P_6=x_6, P_7=x_1+x_2, P_8=x_3+x_4, P_9=x_4+x_5, P_{10}=x_3+x_4+x_5,
\]
and 1 exponential factor, $P_{11}=exp(x_1+x_2+\cdots+x_6)$. The matrix of coefficients of the affine cofactors is
\[
M=\begin{pmatrix}
b  & 0 & 1 & 1 & 1 & 1 & 1 
\\
 b  & -1 & 0 & 1 & 1 & 1 & 1 
\\
 b +d  & -1 & -1 & 0 & 1 & 1 & 1 
\\
 b +d  & -1 & -1 & -1 & 0 & 1 & 1 
\\
 b +d  & -1 & -1 & -1 & -1 & 0 & 1 
\\
 d  & -1 & -1 & -1 & -1 & -1 & 0 
\\
 b  & 0 & 0 & 1 & 1 & 1 & 1 
\\
 b +d  & -1 & -1 & 0 & 0 & 1 & 1 
\\
 b +d  & -1 & -1 & -1 & 0 & 0 & 1 
\\
 b +d  & -1 & -1 & 0 & 0 & 0 & 1 
\\
 0 & b  & b  & b +d  & b +d  & b +d  & d  
\end{pmatrix}.
\]
A basis for $\text{Ker}(M^T)$ is given by $\{(0, 0, -d, d, -d, b +d, 0, 0, 0, 0, 1)$, $(0, 0, -1, 1, -1, 0, 0, 0, 0, 1, 0)$, $(0, 0, -1, 0, -1, 0, 0,$ $1, 1, 0, 0)$, $(0, 0, -1, 1, -1, 1, 1, 0
, 0, 0, 0)\}$. These give rise to four integrals
\begin{equation} \label{ints}
\begin{split}
K_1=x_{3}^{-d} x_{4}^{d} x_{5}^{-d} x_{6}^{b +d} {\mathrm e}^{x_{1}+x_{2}+\cdots+x_{6}}
,\qquad &
K_2=\frac{x_{4} \left(x_{3}+x_{4}+x_{5}\right)}{x_{3} x_{5}}, \\
K_3= 
\frac{\left(x_{3}+x_{4}\right) \left(x_{4}+x_{5}\right)}{x_{3} x_{5}}
,\qquad
& K_4=\frac{x_{4} x_{6} \left(x_{1}+x_{2}\right)}{x_{3} x_{5}},
\end{split}
\end{equation}
of which 3 are functionally independent (we have $K_3-K_2=1$). Note that $K_1=\exp(H_\text{\ref{E3}})$, $K_2=I_{3,5}$, and $K_4=P_{1,2}Q_{3,6}$.
It can be checked, using Lemma \ref{PA}, that the three integrals Poisson commute. This is in accordance with the general results that the class \ref{E3} is Liouville integrable, and admits
$n-3$ functionally independent integrals, cf. section \ref{LSI}.
\end{example}

In the above example, we have seen that two of the three independent integrals we found were known a-priori, one related to the Hamiltonian and a local one. What about the third, could we have predicted this one as well? Moreover, the fact that $K_1$ Poisson commutes with the other integrals is clear, but can we also understand why $K_4$ would commute with $K_2$? The answer to both these questions is affirmative, according to the following theorem.

\begin{theorem}[Contraction theorem] \label{ContThm}
\noindent Consider an inhomogeneous Lotka-Volterra system of the form \eqref{FGLV} with $r_i=r$, for $i=l,l+1,\dots,m$. The new variables
\[
y_i=\begin{cases}
x_i & i<l,\\
P_{l,m} & i=l,\\
x_{i-l+m} &i>l.
\end{cases}
\]
satisfy an inhomogeneous LV system of the form \eqref{FGLV} in dimension $n+l-m$ with
\[
\tilde{r}_i=\begin{cases}
r_i & i<l,\\
r & i=l,\\
r_{i-l+m} &i>l.
\end{cases}
\]
If $F(y)$ is an integral of the $y$-system, then $F(y(x))$ is an integral of the $x$-system. If $F(y)$ and $G(y)$ are commuting integrals of the $y$-system, the integrals also commute as functions of $x$. Moreover, for any function $J(y)$, if $K(x_l,x_{l+1},\ldots,x_m)$ is a local homogeneous integral of weight 0 of the $x$-system, then $\{J(y(x)),K(x)\}=0$.
\end{theorem}
\begin{proof}
\noindent To establish that the $y$-system is of the form \eqref{FGLV}, the main thing we need to prove is $\dot{y}_l = y_l(b - P_{1,l} + P_{m+1,n})$. This goes as follows:
\begin{align*}
\dot{y}_l &= \sum_{i=l}^m \dot{x}_i\\
&= \sum_{i=l}^m x_i(b-P_{1,i-1}+P_{i+1,n})\\
&= \sum_{i=l}^m x_i(b-P_{1,l-1}-P_{l,i-1}+P_{i+1,m}+P_{m+1,n})\\
&= y_l (b-P_{1,l-1}+P_{m+1,n})+Z,
\end{align*}
where
$
Z=\sum_{i=l}^m x_i(-P_{l,i-1}+P_{i+1,m})=\sum_{i,j=l}^m x_iA_{i,j}x_j=0
$
as $A$ is antisymmetric. If $F(y)$ is an integral of the $y$-system, then 
\[
\dot{x}\cdot \nabla F(y(x)) = \sum_{i<l,i>m} \dot{x}_i\frac{\partial F(y(x))}{\partial x_i} + \sum_{l\leq i\leq m} \dot{x}_i\frac{\partial F(y(x))}{\partial x_i}
= \sum_{i<l,i>l} \dot{y}_i \frac{\partial F(y)}{\partial y_i} + \dot{y}_l \frac{\partial F(y)}{\partial y_l}=\dot{y}\cdot \nabla F(y) = 0,
\]
as by the chain rule we have, for $l\leq i\leq m$, 
\[
\frac{\partial F(y(x))}{\partial x_i} = \frac{\partial F(y)}{\partial y_l}\frac{\partial y_l}{\partial x_i} = \frac{\partial F(y)}{\partial y_l}.
\]
For commuting integrals $F(y)$ and $G(y)$ we have
\begin{align*}
\{F(y(x)),G(y(x))\}&=\sum_{i,j<l,i,j>m} \frac{\partial F(y(x))}{\partial x_i} \frac{\partial G(y(x))}{\partial x_j} \{x_i,x_j\} +
\sum_{i<l,i>m,l\leq j\leq m} \frac{\partial F(y(x))}{\partial x_i} \frac{\partial G(y(x))}{\partial x_j} \{x_i,x_j\}\\
&\ \ \ +\sum_{l\leq i\leq m,j<l,j>m,} \frac{\partial F(y(x))}{\partial x_i} \frac{\partial G(y(x))}{\partial x_j} \{x_i,x_j\}+
\sum_{l\leq i,j\leq m} \frac{\partial F(y(x))}{\partial x_i} \frac{\partial G(y(x))}{\partial x_j} \{x_i,x_j\}\\
&=\sum_{i,j<l,i,j>l} \frac{\partial F(y)}{\partial y_i} \frac{\partial G(y)}{\partial y_j} \{y_i,y_j\} +
\sum_{i<l,i>l} \frac{\partial F(y)}{\partial y_i} \frac{\partial G(y)}{\partial y_l} \{y_i,y_l\}+\sum_{j<l,j>l,} \frac{\partial F(y)}{\partial y_l} \frac{\partial G(y)}{\partial y_j} \{y_l,y_j\}\\
&\ \ \ +
\frac{\partial F(y)}{\partial y_l} \frac{\partial G(y)}{\partial y_l} \{y_l,y_l\}=\{F(y),G(y)\}=0.
\end{align*}
Finally, if $K(x_l,x_{l+1},\ldots,x_m)$ is a homogeneous local integral of weight 0, then it is also an integral of the homogeneous (and Hamiltonian) LV-system $\dot{x}_i=\{x_i,y_l\}$, $i=l,\ldots,m$. Hence
\begin{align*}
\{F(y(x)),K(x)\}&=\sum_{i<l,i>m,l\leq j\leq m} \frac{\partial F(y(x))}{\partial x_i} \frac{\partial K(x)}{\partial x_j} \{x_i,x_j\} + \sum_{l\leq i,j\leq m} \frac{\partial F(y(x))}{\partial x_i} \frac{\partial K(x)}{\partial x_j} \{x_i,x_j\} \\
&=\sum_{i<l,i>m} (-1)^{[i>m]}x_i\frac{\partial F(y(x))}{\partial x_i} \sum_{l\leq j\leq m}\frac{\partial K(x)}{\partial x_j}x_j
+ \frac{\partial F(y)}{\partial y_l}\sum_{l\leq j\leq m} \frac{\partial K(x)}{\partial x_j} \{y_l,x_j\} \\
&=(\sum_{i<l,i>m} 0) + 0.
\end{align*}
\end{proof}

\begin{example}
We consider the $k=2,l=3,m=1$ case again. Armed with the contraction theorem, Theorem \ref{ContThm}, we can see where the integral $K_4$ comes from. Indeed, the odd-dimensional $y$-system with $k=1,l=3,m=1$ admits the Casimir $Q_{1,5}$ as integral, and this gives rise to $K_4$ by setting $(y_1,y_2,y_3,y_4,y_5)=(x_1+x_2,x_3,x_4,x_5,x_6)$. However, we cannot use the contraction Theorem to prove that $K_4$ Poisson commutes with $K_2$ (or $K_3$) directly. To do this, we need to define
$K_5=K_2/K_4=P_{1,3}Q_{6,6}/P_{1,2}$ which can be contracted to the case $k=2,l=1,m=1$. It now also follows that $\{K_5,K_2\}=0$, and hence that $\{K_4,K_2\}=0$. Thus, the set $\{H_\text{\ref{E3}},K_2,K_4\}$ is pairwise commuting and the system is Liouville integrable, but not Liouville superintegrable.
\end{example}

We note that while, for the purpose of this paper, the contraction theorem is formulated for systems of the form \eqref{FGLV}, it holds more generally. We give an example in Appendix B.

\section{Integrability (results)} \label{SecI}
In this section, we provide explicit sets of sufficiently many functional independent integrals for the types of integrability described in the introduction. We provide the proofs in the next section.

In section \ref{LIS}, we state that the O and E systems are Liouville integrable. In section \ref{MSI}, we list the Liouville maximally superintegrable cases. These are: \ref{E1} with $n=2$, \ref{O1} with $n=3$, and \ref{O2} with $k=1$. The cases that are nonholonomically integrable are listed in  
section \ref{NHI}: \ref{E1} with $km>1$, \ref{E2}, \ref{E3} with $l=m=1$, \ref{E4} with $k=m=1$, \ref{O1} with two of $k,l,m$ equal to 1, \ref{O2} with $k>1$, \ref{O3} with $l=1$, and \ref{N1}. In section \ref{LSI}, we list the Liouville superintegrable cases with $n-3$ integrals: \ref{E3} with $lm>1$, \ref{E4} with $k>1$ or $m>1$, \ref{E5}, \ref{O1} with at most one of $k,l,m$ equal to 1, \ref{O3} with $l>1$. And finally, in section \ref{NLSI}, we list integrable-or-not cases with $n-3$ integrals: \ref{M1}, \ref{M2}.          

\subsection{Liouville integrable systems} \label{LIS}
We denote sets of commuting local integrals by
\[
V=\{I_{1,2i+1}\}_{i=1}^{\lfloor \frac{k-1}{2}\rfloor} \cup \{I_{k+1,k+2i+1}\}_{i=1}^{\lfloor \frac{l-1}{2}\rfloor} \cup \{I_{k+l+1,k+l+2i+1}\}_{i=1}^{\lfloor \frac{m-1}{2}\rfloor}.
\] 
The numbers of integrals contained in $V$ for the various cases are given in Table \ref{Ta}.
\begin{table}[h]
\begin{center}
\begin{tabular}{r|c|c|c|c|c}
\phantom{$\dfrac12$} $n$-dim. LV-system & E1 & E2, E3, E4 & E5 & O1 & O3\\
\hline
\phantom{$\dfrac12$} Size of the set $V$ & $\frac{n}{2}-1$ & $\frac{n}{2}-2$ & $\frac{n}{2}-3$ & $\frac{n-1}{2}-1$ & $\frac{n-1}{2}-2$ 
\end{tabular}
\caption{\label{Ta} The number of functionally independent commuting local integrals.}
\end{center}
\end{table}

\begin{itemize}
\item[\eqref{E1}]
The set of $n/2$ integrals
$
\{H_\text{\ref{E1}}\} \cup V,
$
is functionally independent and commuting. 
\item[\eqref{E2}]
The set of $n/2$ integrals
\[
\{H_\text{\ref{E2}}, \frac{Q_{1,k}^bP_{1,k}^dQ_{k+1,n}^d}{P_{k+1,n}^b}\} \cup V
\]
is functionally independent and commuting.

Note that, for all values of the parameters $b$ and $d$ (with corresponding $c=b+d$), the systems \ref{E2} share the same set of $n/2-2$ parameter-independent commuting integrals $V$. Similar statements hold for \ref{E1}, \ref{E3}, \ref{E4}, \ref{E5}, \ref{O1}, \ref{O2}, \ref{O3}. 
\item[\eqref{E3}] The set of $n/2$ integrals
\[
\{H_\text{\ref{E3}}, \frac{P_{1,k}P_{k+l+1,n}}{P_{k+1,k+l}}\} \cup V
\]
is functionally independent and commuting.
\item[\eqref{E4}] The set of $n/2$ integrals
\[
\{ H_\text{\ref{E4}}, \frac{P_{1,k}P_{k+l+1,n}}{P_{k+1,k+l}}\} \cup V
\]
is functionally independent and commuting.
\item[\eqref{E5}] The set of $n/2$ integrals
\[
\{H_\text{\ref{E5}}, \frac{P_{1,k}P_{k+l+1,n}}{P_{k+1,k+l}},
\frac{P_{1,k}^dQ_{1,k+l}^bQ_{k+1,n}^d}{P_{k+l+1,n}^b}
\} \cup V
\]
is functionally independent and commuting.
\item[\eqref{O1}] The set of $(n-1)/2$ integrals and 1 Casimir
$
\{ H_\text{\ref{O1}}, Q_{1,n}\} \cup V
$
is functionally independent and commuting.
\item[\eqref{O2}] The set of $(n-1)/2$ integrals and 1 Casimir
$
\{H_\text{\ref{O2}}, Q_{1,n}\} \cup \{I_{k+1,k+2i+1}\}_{i=1}^{\lfloor \frac{l+m-1}{2}\rfloor}
$
is functionally independent and commuting. Note, we can have $l=0$ or $m=0$.

\item[\eqref{O3}] The set of $(n-1)/2$ integrals and 1 Casimir
\[
\{H_\text{\ref{O3}}, \frac{P_{1,k}P_{k+l+1,n}}{P_{k+1,k+l}}, Q_{1,n}\} \cup V
\]
is functionally independent and commuting.   
\end{itemize}
\subsection{Liouville maximally superintegrable systems} \label{MSI}
\begin{itemize}
\item[\eqref{E1}] When $k=m=1$ the 2-dimensional system has one integral, $H_\text{\ref{E1}}$.
\item[\eqref{O1}] When $k=l=m=1$, the 3-dimensional system has two functionally independent integrals, $H_\text{\ref{O1}}$ and the Casimir $Q_{1,3}$.
\item[\eqref{O2}] When $k=1$ the set of $n-1$ integrals
$\{H_\text{\ref{O2}}, Q_{1,n}\} \cup \{I_{i,i+2}\}_{i=1}^{k-2} \cup \{I_{i+1,i+3}\}_{i=1}^{l+m-2}$
is functionally independent.
\end{itemize}
      
\subsection{Nonholonomically integrable systems (quasi-superintegrable)} \label{NHI}
For any $k,l,m$ we denote
\[
X=X(k,l,m)= \{I_{i,i+2}\}_{i=1}^{k-2} \cup \{I_{k+i,k+i+2}\}_{i=1}^{l-2} \cup \{I_{k+l+i,k+l+i+2}\}_{i=1}^{m-2}.
\]
If 1 appears $j$ times in the sequence $k,l,m$ then $X$ contains $n-6 + j$ integrals when $l\neq 0$, or $n-4 + j$ integrals when $l=0$. In the sequel, we will also use the Iverson bracket:
\[
[S]=\begin{cases} 1 & S \text{ is true} \\ 0 & S \text{ is false}. \end{cases}
\] 
\begin{itemize}
\item[\eqref{E1}] When $km>1$ the set of $n-2$ integrals
\[
\{H_\text{\ref{E1}},\frac{[k>1,m>1]P_{1,2}^b Q_{3,n}^b Q_{1,n-2}^d}{P_{n-1,n}^d}\} \cup X
\]
is functionally independent.
\item[\eqref{E2}]
The set of $n-2$ integrals
\[
\{H_\text{\ref{E2}}, \frac{Q_{1,k}^bP_{1,k}^dQ_{k+1,n}^d}{P_{k+1,n}^b}\} \cup
X
\]
is functionally independent.
\item[\eqref{E3}] When $l= m= 1$, the set of $n-2$ integrals
$
\{H_\text{\ref{E3}}, \frac{P_{1,k}P_{k+l+1,n}}{P_{k+1,k+l}}\} \cup X
$
is functionally independent.
\item[\eqref{E4}] When $k=m=1$, the set of $n-2$ integrals
$
\{ H_\text{\ref{E4}}, \frac{P_{1,k}P_{k+l+1,n}}{P_{k+1,k+l}}\}\} \cup X
$
is functionally independent.
\item[\eqref{O1}] When 1 appears twice in $k,l,m$, the set of $n-2$ integrals
$
\{ H_\text{\ref{O1}}, Q_{1,n}\} \cup X
$
is functionally independent.

\item[\eqref{O2}] When $k>1$ the set of $n-2$ integrals
\[
\{H_\text{\ref{O2}}, Q_{1,n}\}
\cup \{I_{i,i+2}\}_{i=1}^{k-2} \cup \{I_{k+i,k+i+2}\}_{i=1}^{l+m-2}
\]
is functionally independent.
\item[\eqref{O3}] When $l=1$, the set of $n-2$ integrals
\[
\{H_\text{\ref{O3}}, \frac{P_{1,k}P_{k+2,n}}{x_{k+1}}, Q_{1,n}\} \cup X
\]
is functionally independent.

\item[\eqref{N1}] The set of $n-2$ integrals
\[
\{ \frac{\e^{P_{1,n}}P_{k+1,n}^b}{P_{1,k}^{d}},
[k>1]P_{1,2}^bP_{k+1,n}^bQ_{3,k}^bQ_{1,n}^d\} \cup X
\]
is functionally independent.
\end{itemize}
\subsection{Liouville superintegrable systems with $n-3$ integrals} \label{LSI}
\begin{itemize}
\item[\eqref{E3}] When $lm>1$ the set of $n-3$ integrals
\[
\{H_\text{\ref{E3}},\frac{P_{1,k}P_{k+l+1,n}}{P_{k+1,k+l}}\},
\frac{[l>1,m>1]P_{n-1,n}^dQ_{1,k}^{b+d}}{Q_{1,n-2}^dP_{k+1,k+2}^{b+d}Q_{k+3,n}^{b+d}}
\} \cup X
\]
is functionally independent.
\item[\eqref{E4}] When $k>1$, the set of $n-3$ integrals
\[
\{H_\text{\ref{E4}}, \frac{P_{1,k}P_{k+l+1,n}}{P_{k+1,k+l}},
\frac{[m>1]P_{1,2}^bQ_{3,n}^b Q_{1,n-2}^d}{P_{n-1,n}^d}\}
\cup X
\]
is functionally independent.

\noindent The case $k=1,m>1$ is equivalent to the case $k>1,m=1$, included above.

\item[\eqref{E5}] The set of $n-3$ integrals
\[
\{H_\text{\ref{E5}}, \frac{P_{1,k}P_{k+l+1,n}}{P_{k+1,k+l}},
\frac{P_{1,k}^dQ_{1,k+l}^bQ_{k+1,n}^d}{P_{k+l+1,n}^b}\}
\cup X
\]
is functionally independent.

\item[\eqref{O1}] When 1 appears at most once in $k,l,m$ the set of $n-3$ integrals
\[
\{ H_\text{\ref{O1}}, Q_{1,n},
\frac{[k>1,l>1,m>1]P_{k-1,k}^bQ_{k+1,k+2}^bQ_{k-1,k}^{b+d}P_{n-1,n}^d}{P_{k+1,k+2}^{b+d}Q_{1,k-2}^dQ_{k+3,n-2}^d}
\}
\cup X
\]
is functionally independent.

\item[\eqref{O3}] When $l>1$ the set of $n-3$ integrals
\[
\{H_\text{\ref{O3}}, \frac{P_{1,k}P_{k+l+1,n}}{P_{k+1,k+l}}, Q_{1,n}\} \cup X
\]
is functionally independent.
\end{itemize}

\subsection{Integrable-or-not systems with $n-3$ integrals} \label{NLSI}
\begin{itemize}
\item[\eqref{M1}] The set of $n-3$ integrals
\[
\{
\frac{\e^{P_{1,n}}P_{k+1,k+l}^b}{P_{1,k}^{b+d}},
\frac{P_{1,k}P_{k+l+1,n}}{P_{k+1,k+l}},
\frac{[k>1]Q_{k-1,k+l}^bP_{k-1,k}^b}{Q_{1,n}^dP_{k+1,k+l}^b}
\} 
\cup X
\]
is functionally independent.

\item[\eqref{M2}] The set of $n-3$ integrals
\[
\{
\frac{\e^{P_{1,n}}P_{k+1,k+l}^b}{P_{1,k}^{b+d}},
\frac{P_{1,k}P_{k+l+1,n}}{P_{k+1,k+l}},
\frac{[l>1]P_{1,k}^{b+d}Q_{1,n}^b}{P_{k+1,k+2}^{b+d}Q_{k+3,n}^{b+d}}
\}
\cup X
\]
is functionally independent.
\end{itemize}

\section{Integrability (proof of results)} \label{SecP}
We provide proof that the sets of integrals provided in the previous section are functionally independent and, in the Liouville integrable cases, sufficiently many of them commute.
\subsection{The Liouville integrability of the E and O cases} \label{pli}
\begin{itemize}
\item[\eqref{E1}] We know from the proof of Theorem 4, that each set in the union
$V=\{I_{1,2i+1}\}_{i=1}^{\lfloor \frac{k-1}{2}\rfloor} \cup \{I_{k+1,k+2i+1}\}_{i=1}^{\lfloor \frac{m-1}{2}\rfloor}$ is pairwise commuting. We now show that $V$ is pairwise commuting. Consider the bracket, with $i\leq k$ odd and $j>k+1$ even,
\begin{align*}
\{I_{1,i},I_{k+1,j}\}&=
\{P_{1,i},P_{k+1,j}\}Q_{1,i}Q_{k+1,j}+
\{P_{1,i},Q_{k+1,j}\}Q_{1,i}P_{k+1,j}+
\{Q_{1,i},P_{k+1,j}\}P_{1,i}Q_{k+1,j}\\
&\ \ \ +
\{Q_{1,i},Q_{k+1,j}\}P_{1,i}P_{k+1,j}.
\end{align*}
It follows from Lemma \ref{PA} that the first and the last terms equal $P_{1,i}P_{k+1,j}Q_{1,i}Q_{k+1,j}$, and the middle two terms equal $-P_{1,i}P_{k+1,j}Q_{1,i}Q_{k+1,j}$. Clearly, the integrals can be ordered so that each integral depends on one or more new variables.
\item[\eqref{E2}] The proof that $V$ is pairwise commuting is similar to the \ref{E1} case, but with $j>k+1$ odd. We need to show that $K= Q_{1,k}^bP_{1,k}^dQ_{k+1,n}^d/P_{k+1,n}^b$ is a commuting integral. This can be done by a two-step contraction to the $k=m=1$ case, with Hamiltonian $y_1+y_2-d\ln(y_1)+b\ln(y_2)$. We get the integral $K$ by substitution of $(y_1,y_2)=(P_{1,k},P_{k+1,n})$, subtracting the result from $H_\text{\ref{E2}}$ and exponentiation. By the contraction theorem it then commutes with all local integrals. To prove independence, we notice that the local integrals only depend on the variable
    $x_1,\ldots,x_{k-1},x_{k+1},\ldots,x_{n-1}$, and that at $x_1=\cdots=x_n=1$ we have
\[
\frac{D(H_{\ref{E2}},K)}{D(x_k,x_n)}=    \begin{pmatrix}
    1+b  & 1+d  
\\
 \frac{k^{d}}{m^{b}}\left(b+\frac{d}{k}\right)& \frac{k^{d} }{m^{b}}\left(d-\frac{b}{m} \right)
\end{pmatrix}.
\]
\item[\eqref{E3}] That $V$ is pairwise commuting follows from the above considerations. The function $K=P_{1,k}P_{k+l+1,n}/P_{k+1,k+l}$
is obtained, by a 3-step contraction to the case $k=l=m=1$, which admits the Casimir $y_1y_3/y_2$, and commutes with all the local integrals. To prove independence we introduce
$y_i=x_i/x_{i+1}$, $i=k+l+1,\ldots,n$, and note that the local integrals only depend on $x_1,\ldots,x_{k-1},x_{k+1},\ldots,x_{k+l},y_{k+l+1},\ldots,y_{n-1}$ and that at $x_1=\cdots=y_{n-1}=x_n=1$ we have
\[
\frac{D(H_{\ref{E3}},K)}{D(x_k,x_n)}=    \begin{pmatrix}
    1  & m+b+d  
\\
 \frac{m}{l}& \frac{km}{l}
\end{pmatrix}.
\] 
\item[\eqref{E4}] The proof is similar to the above.
\item[\eqref{E5}] We perform a 3-step contraction to the case $k=l=m=1$. This a admits the Casimir $y_1y_3/y_2$, which gives rise to $K=P_{1,k}P_{k+l+1,n}/P_{k+1,k+l}$, and has Hamiltonian $y_1+y_2+y_3-b\ln(y_1/y_2)-d\ln(y_2/y_3)$ to which we add $(b-d)\ln(y_1y_3/y_2))$ to get $T(y_1,y_2,y_3)=y_1+y_2+y_3+b\ln(y_3)-d\ln(y_1)$. The integral $L=P_{1,k}^dQ_{1,k+l}^bQ_{k+1,n}^d/P_{k+l+1,n}^b$ is obtained by as $\exp(H_\text{\ref{E5}}-T(P_{1,k},P_{k+1,k+l},P_{k+l+1,n}))$, and it commutes with all local integrals as it is a function of integrals that commute with all local integrals.
We do not need to change variables, as the local integrals do not depend on $x_k$, $x_{k+l}$ and $x_n$. The relevant Jacobian matrix
\[
\frac{D(H_{\ref{E5}},K,L)}{D(x_k,x_{k+l},x_n)}=    \begin{pmatrix}
1+b  & 1+b +d  & 1+d  
\\
 \frac{k^{d-1}}{m^{b}} \left(b k +d \right) & \frac{k^{d}}{m^{b}}\left(b +d \right) & \frac{k^{d}}{m^{b+1}} \left(d m - b \right) 
\\[1mm]
 \frac{m}{l} & -\frac{k m}{l^{2}} & \frac{k}{l} 
\end{pmatrix}
\]
has full rank.
\item[\eqref{O1}] All integrals commute. The local integrals depend on all the variables, so we need to introduce two sets of homogeneous variables, e.g. $y_i=x_i/x_{i+1}$ for $i=1,\ldots,k-1$ and for $i=k+1,\ldots,k+l-1$, and then check that the Jacobian of $H_\text{\ref{O1}},Q_{1,n}$ with respect to $x_k,x_{k+l}$ has full rank.
\item[\eqref{O2}] All integrals commute. The local integrals do not depend on $x_n$. We only need one set of homogeneous variables, $y_i=x_i/x_{i+1}$ for $i=1,\ldots,k-1$ and check that the Jacobian of $H_\text{\ref{O2}},Q_{1,n}$ with respect to $x_k,x_{n}$ has full rank.
\item[\eqref{O3}] The integral $K=P_{1,k}P_{k+l+1,n}/P_{k+1,k+l}$ comes from the Casimir of the 3-step contracted system $k=l=m=1$ and hence commutes with all other integrals. The local integrals do not depend on $x_k,x_n$, one needs to introduce one set of homogeneous variables, $y_i=x_i/x_{i+1}$ for $i=k+1,\ldots,k+l-1$ and check that the Jacobian of $H_\text{\ref{O3}},K,Q_{1,n}$ with respect to $x_k,x_{k+l},x_{n}$ has full rank. 
\end{itemize}

\subsection{The maximally superintegrable cases}
For the lowest dimensional systems \ref{E1} and \ref{O1} the results are known. For the case \ref{O2} with $k=1$, it is clear that the integrals Poisson-commute. To prove independence, we introduce two set of homogeneous variables, $y_i=x_i/x_{i+1}$ for $i=1,\ldots,k-1$, and for $i=k+l+1,\ldots,n$, which we use, together with $x_k,x_n$ to express all integrals. The local integrals can be ordered so that each integral depends on a new variable, and they don't depend on $x_k,x_n$. One can check that the Jacobian of $H_\text{\ref{O2}},Q_{1,n}$ with respect to $x_k,x_{n}$ has full rank.
\subsection{The nonholonomically integrable cases}
To prove that the LV-systems are nonholonomically integrable includes showing that they are measure preserving. We do this for all cases simultaneously in the following proposition.
\begin{proposition} \label{MP}
Lotka-Volterra equations (\ref{LVgen}) with $a_{ii}=0$ are measure preserving with density $M=\left(\prod_{i=1}^n x_i\right)^{-1}$.
\end{proposition}
\begin{proof}
Denoting the right hand side of (\ref{LVgen}) by $f_i$, for all $i$ the expression $Mf_i$ does not depend on $x_i$, hence $\sum_{i=1}^n \frac{\partial Mf_i}{\partial x_i} = 0$.
\end{proof}

\noindent
We now prove the existence and independence of $n-2$ integrals for each of the cases mentioned in section \ref{NHI}.
\begin{itemize}
\item[\eqref{E1}] When $k=1$ or $m=1$ there are $n-3$ integrals in the set $X$, the local
integrals do not depend on $x_1$ or $x_n$ respectively, and $H_\text{\ref{E1}}$ does depend on that variable. When $k>1,m>1$, we have an additional nonlocal integral $K=P_{1,2}^b Q_{3,n}^b Q_{1,n-2}^d/P_{n-1,n}^d$. One can prove that $K$ is an integral by using the cofactors of its constituents\footnote{The cofactors of $P_{1_2}$, $Q_{3,n}$, $Q_{1,n-2}$, $P_{n-1,n}$ are $P_{3,n}+b$, $-P_{3,n}-b+d$, $-P_{1,n-2}-b+d$, $-P_{1,n-2}+d$ respectively, cf. \eqref{CP} and \eqref{CQ}.}. However, we prefer to understand whether and how it arises from the contraction theorem. This integral can be obtained as follows. Let $T$ be the integral obtained from the Hamiltonian of the 2-step contracted system $k=m=1$, and let $R$ be the integral obtained from a 2-step contraction to the $k-1,m-1$ \ref{E2} case, by substituting $y_1=x_1+x_2$, $y_i=x_{i+1}$, $i=2,\ldots,n-1$ and $y_{n-2}=x_{n-1}+x_n$ in $R(y)=P_{1,k-1}^dQ_{k,n-2}^dQ_{1,k-1}^b/P_{k,n-2}^b$. Then $K=\exp(H_\text{\ref{E1}}-T)/R$.
To prove independence, one introduces two sets of homogeneous variables,
$y_i=x_i/x_{i+1}$ for $i=1,\ldots,k-1$, and for $i=k+1,\ldots,n-1$. The local integrals do not depend on $x_k,x_n$ and the Jacobian of $H_\text{\ref{E1}},K$ with respect to $x_k,x_n$ has full rank.
\item[\eqref{E2}] The integral $K= Q_{1,k}^bP_{1,k}^dQ_{k+1,n}^d/P_{k+1,n}^b$ was derived in section \ref{pli}.
To prove independence, one introduces two sets of homogeneous variables,
$y_i=x_i/x_{i+1}$ for $i=1,\ldots,k-1$, and for $i=k+1,\ldots,n-1$. The local integrals do not depend on $x_k,x_n$ and the Jacobian of $H_\text{\ref{E1}},K$ with respect to $x_k,x_n$ has full rank.
\item[\eqref{E3}] 
The integral $K=P_{1,k}P_{k+l+1,n}/P_{k+1,k+l}$ was derived in section \ref{pli}.
The Jacobian of $H_\text{\ref{E3}},K$ with respect to $x_{n-1},x_n$ has full rank.    
\item[\eqref{E4}] Similar to the previous case. With $K=P_{1,k}P_{k+l+1,n}/P_{k+1,k+l}$, the Jacobian of $H_\text{\ref{E4}},K$ with respect to $x_{1},x_n$ has full rank.
\item[\eqref{O1}] Let
$
(i,j)=\begin{cases}
(k,k+l) & k=l=1,\\
(k,n) & k=m=1,\\
(k+l,n) & l=m=1.
\end{cases}
$
The Jacobian of $H_\text{\ref{E5}},Q_{1,n}$ with respect to $x_{i},x_j$ has full rank.
\item[\eqref{O2}] To prove independence, one introduces two sets of homogeneous variables,
$y_i=x_i/x_{i+1}$ for $i=1,\ldots,k-1$, and for $i=k+1,\ldots,n-1$. The local integrals do not depend on $x_k,x_n$ and the Jacobian of $H_\text{\ref{O2}},Q(1,n)$ with respect to $x_k,x_n$ has full rank.
\item[\eqref{O3}] To prove independence, one introduces two sets of homogeneous variables,
$y_i=x_i/x_{i+1}$ for $i=1,\ldots,k-1$, and for $i=k+2,\ldots,n-1$. The local integrals do not depend on $x_k,x_{k+1},x_n$ and the Jacobian of $H_\text{\ref{O3}},P_{1,k}P_{k+2,n}/x_{k+1}$, $Q(1,n)$ with respect to $x_k,x_{k+1},x_n$ has full rank.
\item[\eqref{N1}] The integral $E=\e^{P_{1,n}}P_{k+1,n}^b/P_{1,k}^{d}$ is obtained from the Hamiltonian of the contracted \ref{E1} system $k=m=1$. For $k=1$ the set is functionally independent because the local integrals do not depend on $x_1$, and $E$ does. For $k>1$, we also have the nonlocal integral $K=P_{1,2}^bP_{k+1,n}^bQ_{3,k}^bQ_{1,n}^d$.
    This integral can be obtained from the contracted \ref{E2} system with $\tilde{k}=k-1,\tilde{m}=m$. Let $R(y)=P_{k,n-1}^b/(Q_{1,k-1}^bP_{1,k-1}^dQ_{k,n-1}^d)$. Substitute $y_1=x_1+x_2$, $y_i=x_{i+1}$, $i=2,\ldots,n-1$. Then $K=R(y(x))I_{1,3}^d$. To prove independence, one introduces two sets of homogeneous variables,
$y_i=x_i/x_{i+1}$ for $i=1,\ldots,k-1$, and for $i=k+1,\ldots,n-1$, and verifies that the Jacobian of $E,K$ with respect to $x_k,x_n$ has full rank.
    \end{itemize}
\subsection{The Liouville superintegrable cases with $n-3$ integrals}
\begin{itemize}
\item[\eqref{E3}] Let $J=P_{1,k}P_{k+l+1,n}/P_{k+1,k+l}$ be the integral obtained from the Casimir of the contracted system $k=l=m=1$. When either $l=1$ or $m=1$ it suffices to introduce 1 set of homogeneous variables, $y_i=x_i/x_{i+1}$ for $i=1,\ldots,k-1$ and then verify that the Jacobian of $H_\text{\ref{E3}},J$ with respect to either $x_k,x_{k+1}$ or
    $x_k,x_n$ has full rank. When $l>1,m>1$ we have the additional nonlocal integral $K=
(P_{n-1,n}/Q_{1,n-2})^d/(P_{k+1,k+2}Q_{k+3,n}/Q_{1,k})^{b+d}$. This integral is obtained from a contraction to an \ref{E5} system with $\tilde{k}=k,\tilde{l}=l-1,\tilde{m}=m-1$.
Let $T(y)=P_{1,k}^dQ_{k+1,n-2}^dQ_{1,k+l-1}^b/P_{k+l,n-2}^b$, then $T(x)$ is obtained by $y_{k+1}=x_{k+1}+x_{k+2}$, $y_i=x_{i+1}$, $i=k+2,\ldots,n-3$, $y_{n-2}=x_{n-1}+x_{n}$, and we have $K=TJ^bI_{k+1,k+l}^b/R^{b+d}$, where $R=P_{1,k}Q_{k+1,n}$ is a nonlocal integral obtained from the Casimir of the contracted \ref{O1} system with $\tilde{k}=1,\tilde{l}=l,\tilde{m}=m$. We could have taken $T$ as the additional integral (but not $R$), however, the expression $K$ is nicer. 

For independence one needs to introduce 3 sets of homogeneous variables $y_i=x_i/x_{i+1}$, with $i=1,\ldots,k-1$, $i=k+1,\ldots,k+l-1$ and $i=k+l+1,\ldots,n-1$, and verify that the Jacobian of $H_\text{\ref{E3}},J,K$ with respect to $x_k,x_{k+l},x_n$ has full rank.
\item[\eqref{E4}] Let $J=P_{1,k}P_{k+l+1,n}/P_{k+1,k+l}$. When $m=1$ the local integrals do not depend on $x_n$. One introduces 1 set of homogeneous variables, e.g. $y_i=x_i/x_{i+1}$, with $i=1,\ldots,k-1$ and verifies  that the Jacobian of $H_\text{\ref{E4}},J$ with respect to $x_k,x_n$ has full rank. When $k>1,m>1$ we have an additional nonlocal integral $K=P_{1,2}^bQ_{3,n}^bQ_{1,n-2}^d/P_{n-1,n}^d$.
Similar to the previous case it can be obtained from a contraction to an \ref{E5} system, this time with $\tilde{k}=k-1,\tilde{l}=l,\tilde{m}=m-1$.
Let $T(y)=P_{1,k-1}^dQ_{k,n-2}^dQ_{1,k+l-1}^b/P_{k+l,n-2}^b$, then $T(x)$ is obtained by $y_{1}=x_{1}+x_{2}$, $y_i=x_{i+1}$, $i=2,\ldots,n-3$, $y_{n-2}=x_{n-1}+x_{n}$, and $K=I_{1,k}^d/I_{k+l+1,n}^b/T$. The proof of independence is similar to the \ref{E3} $l>1,m>1$ case.
\item[\eqref{E5}] The proof of independence is similar to the \ref{E3} $l>1,m>1$ case.
\item[\eqref{O1}] When 1 appears at once in $k,l,m$ one introduces 1 set of inhomogeneous variables, e.g, $y_i=x_i/x_{i+1}$, with $i=1,\ldots,k-1$ if $m=1$, and verifies that the Jacobian of $H_\text{\ref{O1}},Q(1,n)$ with respect to $x_k,x_n$ has full rank.
    We have an additional nonlocal integral $K=P_{k-1,k}^b P_{n-1,n}^d Q_{k-1,k+2}^b/P_{k+1,k+2}^{b+d}/Q_{1,k}^d/Q_{k+3,n-2}^d$,
    when none of $k,l,m$ is equal to 1. This integral can be obtained by contraction from the \ref{E5} case with $\tilde{k}=k-1,\tilde{l}=l-1,\tilde{m}=m-1$. We define $T(y)=P_{1,k-1}^dQ_{k,n-3}^dQ_{1,k+l-2}^b/P_{k+l-1,n-3}^b$ and take $y_{k-1}=x_{k-1}+x_k$, $y_k=x_{k+1}+x_{k+2}$, $y_i=x_{i+2}$, $i=k+1,\ldots,n-4$, $y_{n-3}=x_{n-1}+x_n$. Then we have
    $K=TI_{k+l+1,n}^b/I_{1,k}^d/Q_{1,n}^b$. The proof of independence is similar to the \ref{E3} $l>1,m>1$ case.

\item[\eqref{O3}] The proof of independence is similar to the \ref{E3} $l>1,m>1$ case.
\end{itemize}

\begin{itemize} 
\item[\eqref{M1}] The integral $J= P_{1,k}P_{k+l+1,n}/P_{k+1,k+l}$ is obtained from contraction to the $n=3$ \ref{O1} system. The integral
    $K=\e^{P_{1,n}}P_{k+1,k+l}^b/P_{1,k}^{b+d}$ is obtained from the exponentiated Hamiltonian $T(y)=\exp(P_{1,m+2})y_2^b/y_1^{b+d}$ of the contracted system with $\tilde{k}=1,\tilde{l}=1,\tilde{m}=m$ (equiv. to \ref{E3}). For $k>1$ we obtain the additional nonlocal integral $L=(Q_{k-1,k+l}P_{k-1,k}/P_{k+1,k+l})^b/Q_{1,n}^d$
from a contraction to an \ref{E5} system with $\tilde{k}=k-1,\tilde{l}=l,\tilde{m}=m$, as follows. Define $G(y)=P_{1,k-1}^dQ_{k,n-1}^dQ_{1,k+l-1}^b/P_{k+l,n-1}^b
$ and set $y_{k-1}=x_{k-1}+x_k$, $y_i=x_{i+1}$, $i=k,\ldots,n-1$. Then $L=GJ^b/I_{1,k}^{b+d}$. The proof of independence is similar to the \ref{E3} $l>1,m>1$ case. 
\item[\eqref{M2}] The integrals $J,K$ have the same form as in the \ref{M1} case. For $l>1$, we can obtain the additional nonlocal integral $L=(P_{1,k}/P_{k+1,k+2}/Q_{k+3,n})^{b+d}Q_{1,n}^b
$ from a contraction to an \ref{E5} system with $\tilde{k}=k,\tilde{l}=l-1,\tilde{m}=m$, as follows. Define $G(y)=P_{1,k}^dQ_{k+1,n-1}^dQ_{1,k+l-1}^b/P_{k+l,n-1}^b
$ and set $y_{k+1}=x_{k+1}+x_{k+2}$, $y_i=x_{i+1}$, $i=k+2,\ldots,n-1$. Then $L=GJ^bI_{k+1,k+l}^b$. The proof of independence is similar to the \ref{E3} $l>1,m>1$ case. 
\end{itemize}

\section{Solvability} \label{SecS}
Using the contraction theorem, Theorem \ref{ContThm}, we know that equation \eqref{LVF} can be contracted to the $2D$ integrable Lotka-Volterra system
\begin{equation} \label{lv2}
\begin{split}
\dot{B} &= B(b+D) \\
\dot{D} &= D(d-B),
\end{split}
\end{equation}
when $l=0$, or, when $l\neq0$, to the
$3D$ integrable Lotka-Volterra system (cf. equation no. 4 in Table 1 of \cite{GMRSW})
\begin{equation} \label{lv3}
\begin{split}
\dot{B} &= B(b+C+D) \\
\dot{C} &= C(c-B+D) \\
\dot{D} &= D(d-B-C),
\end{split}
\end{equation}
where $B=P_{1,k}$, $C=P_{k+1,k+l}$, $D = P_{k+l+1,n}$. We express the solutions for $B(t), C(t), D(t)$ in terms of the Lambert W function ($W$) in the following propositions, whose proofs are given in 
Appendix C. We note that the use of the Lambert W function to solve Lotka-Volterra and other biological models is not new, cf. \cite{CGHJK,Leh,Shi} and references therein. 
\begin{proposition} \label{p1}
Let
\[
F(p)=\int^p \frac{1}{bq\left(W(R(q))+1\right)}dq,\quad \text{with } R(q)=\frac{q^{d/b}}{b\e^{(q+b+c_1)/b}}.
\]
The solution to Lotka-Volterra system \eqref{lv2} is, in terms of the inverse function $F^{-1}$, given by $B(t)=F^{-1}(t+c_2)$ and $D(t)=W(R(B(t)))$.
\end{proposition}
\begin{proposition} \label{p2}
Let
\[
G(p)=\int^p \frac{1}{bq\left(W(T(q))+1\right)}dq,\quad \text{with } T(q)=\frac{{\mathrm e}^{-\frac{q +b}{b}} \left(q^{\frac{d}{b}} c_{1} -q^{\frac{b +d}{b}} c_{2} \right)}{b^{2}}.
\]
The solution to Lotka-Volterra system \eqref{lv3} is, in terms of the inverse function $G^{-1}$, given by
\[
B(t)=G^{-1}(t+c_3),\quad
C=\frac{c_{2} \left(bB -B' \right)}{ \left(c_{1} -c_{2} B \right)},\quad
D=-\frac{c_{1} \left(bB -B' \right)}{B \left(c_{1} -c_{2} B \right)}.
\]
\end{proposition}

We now introduce a linear change of variables which provides a separation of variables for equation \eqref{LVF}, i.e. in the transformed system each equation will depend on only one variable (as well as on the solution of \eqref{lv2} or \eqref{lv3}). Define
\[
y_i=\begin{cases}
P_{1,i}, & i=1,\ldots,k,\\
P_{k+1,i}, & i=k+1,\ldots, k+l,\\
P_{k+l+1,i}, & i=l+1,\ldots, k+l+m=n.
\end{cases}
\]
Each of these variables is a Darboux polynomial (and hence the transformed system is of Lotka-Volterra form). Their cofactors are obtained from \eqref{CP} and \cite[Thm 3.1]{KMQ}, cf. the footnote on page 6. We have $\dot{y}_i=y_iC_i$, with
\[
C_i=\begin{cases}
b+P_{i+1,n}, & i=1,\ldots,k,\\
c-P_{1,k}+P_{i+1,n}, & i=k+1,\ldots, k+l,\\
d-P_{1,k+l}+P_{i+1,n}, & i=l+1,\ldots, k+l+m=n,
\end{cases}
\]
which can be rewritten as
\[
C_i=\begin{cases}
b+B+C+D-y_i, & i=1,\ldots,k,\\
c-B+C+D-y_i, & i=k+1,\ldots, k+l,\\
d-B-C+D-y_i, & i=l+1,\ldots, k+l+m=n,
\end{cases}
\]
where $B,C,D$ are known functions of $t$, given in Proposition \ref{p1} or Proposition \ref{p2}. Thus, each equation has the form $
\dot{y} = y(f(t) - y)$,
which can be solved explicitly, for any function $f$, by using an integrating factor. With $\dot{F}(t) = f(t)$,
\begin{align*}
\frac{\dot{y}}{y^2} - \frac{f(t)}{y} = -1
&\implies e^{F(t)}\frac{\dot{y}}{y^2} - e^{F(t)}\frac{f(t)}{y} = -e^{F(t)} \\
&\implies
\frac{d}{dt} \left(-\frac{e^{F(t)}}{y} \right) = - e^{F(t)} \\
&\implies
-\frac{e^{F(t)}}{y(t)} = -\frac{e^{F(0)}}{y(0)} - \int e^{F(t)} dt,
\end{align*}
from which $y(t)$ can be found.

\section{Summary and concluding remarks}
We have introduced Darboux functions $P_{i,j}$ and $Q_{i,j}$ and provided, in Lemma \ref{PA}, the quadratic Poisson algebra they satisfy, with respect to the bracket
$
\{x_i,x_j\}=a_{ij}x_ix_j$, 
where
\begin{equation} \label{iam}
a_{ij}=\begin{cases}1 &i<j \\ 0 &i=j \\-1 &i>j\end{cases}
\end{equation}
which is the interaction matrix of Lotka-Volterra system \eqref{LVF} under study. The Poisson algebra was subsequently used to prove the involutivity of local integrals $I_{i,j}=P_{i,j}Q_{i,j}$ and $I_{k,l}$, when $i\equiv j\equiv k\equiv l \mod 2$ and $[k,l]\subset [i,j]$. The Poisson algebra may become interesting in its own right. One could study the relations needed for the algebra to satisfy the Jacobi identity.

Our Lotka-Volterra system \eqref{LVF} naturally splits up into 11 families, depending on whether the parameters $k,l,m$ are even or odd, and whether $b,c=b+d,d$ vanish or not. In section \ref{SecH}, for the cases E1-E5 and O1-O3, we provided their Hamiltonian functions, in terms of the scalings functions we introduced, cf. Lemma \ref{SC}. The cases N1, M1 and M2 do not seem to be Hamiltonian.

We have provided, for each class of systems E1-E5, O1-O3 and N1, sets of sufficiently many integrals for the various notions of integrability. All Hamiltonian systems (the E and O cases) are Liouville integrable. Furthermore, the systems E1 with $n=2$, O1 with $n=3$ and O2 with $k=1$ are maximally superintegrable ($n-1$ functionally independent integrals). The other E and O systems have either $n-2$ or $n-3$ functionally independent integrals. The non-Hamiltonian N1 systems admit $n-2$ integrals, and, as they are also measure preserving, cf. Proposition \ref{MP}, this implies they are nonholonomically integrable. The main tool which we developed and employed to establish the existence of a sufficient number of (commuting) integrals in each case, is the contraction theorem, Theorem \ref{ContThm}. It states that (commuting) integrals of a system with smaller blocks can be lifted to (commuting) integrals, which then also commute with local integrals present.  

The classes for which we were not able to establish integrability (in the sense of sufficiently many integrals) are M1 and M2. For these systems we (only) found $n-3$ functionally independent integrals. As an example, the 5-component system M1 with $(k,l,m)=(1,2,2)$ reads
\begin{equation}\label{tsja}
\begin{split}
\dot x_{1}&=x_{1} \left(b +x_{2}+x_{3}+x_{4}+x_{5}\right),\\ 
\dot x_{2}&=x_{2} \left(b +d -x_{1}+x_{3}+x_{4}+x_{5}\right),\\
\dot x_{3}&=x_{3} \left(b +d -x_{1}-x_{2}+x_{4}+x_{5}\right),\\
\dot x_{4}&=x_{4} \left(d -x_{1}-x_{2}-x_{3}+x_{5}\right),\\
\dot x_{5}&=x_{5} \left(d -x_{1}-x_{2}-x_{3}-x_{4}\right).
\end{split}
\end{equation}
The system \eqref{tsja} admits two functionally independent integrals
\[
\frac{{\mathrm e}^{x_{1}+x_{2}+x_{3}+x_{4}+x_{5}} \left(x_{2}+x_{3}\right)^{b}}{x_{1}^{b +d}},\quad \frac{x_{1} \left(x_{4}+x_{5}\right)}{x_{2}+x_{3}}.
\]
We cannot exclude the possibility that
a third integral for \eqref{tsja} might perhaps be constructed by a method different than the one considered in the present paper.

In Propositions \ref{p1}, \ref{p2} we have provided explicit solutions, in terms of the Lambert W function, for the 2- and 3-component Lotka-Volterra systems (\ref{lv2}), (\ref{lv3}). We subsequently showed how to integrate each $n$-component Lotka-Volterra equation in the class considered, by separation of variables and using an integrating factor.

We believe, but have not proven, that systems of the form \eqref{LVF}, but with different $r\in\R^n$ (e.g. containing more than $3$ blocks) are not integrable.

We have studied Lotka-Volterra systems with a very specific interaction matrix, given by \eqref{iam}. We have formulated the contraction theorem in this setting. However, as shown in Appendix B, contraction can be applied in more general LV-systems. The full scope of contraction remains to be explored. 

\newpage
\appendix

\section{Proof of Lemma \ref{PA}}
We establish the equalities (\ref{PP}), (\ref{PQ}) and (\ref{QQ}).
\begin{enumerate}
\item[(\ref{PP})] We have $m=\max(i,k)=k$ and $p=\min(j,l)$. We distinguish three cases. The RHS can be identified as
\[
P_{i,j}P_{k,l}-P_{m,p}^2-2P_{m,p}P_{l+1,j}=
\begin{cases}
P_{i,j}P_{k,l} & j<k\\
P_{i,j}P_{j+1,l} + P_{i,k-1}P_{k,j} & k\leq j \leq l \\
P_{k,l}(P_{i,k-1}-P_{l+1,j}) & j>l.
\end{cases}
\]
For the LHS, we have $\{P_{i,j},P_{k,l}\}=(\nabla P_{i,j})^T\Omega\nabla P_{k,l}$ with
\[
(\nabla P_{i,j})_a=\begin{cases}0 & a<i \\ 1 & i\leq a \leq j \\ 0 & a > j,\end{cases} \qquad
\Omega_{a,b}=\begin{cases}-x_ax_b & a>b \\ 0 & a=b \\ x_ax_b & a<b. \end{cases}
\]
Hence
\begin{equation}\label{PO}
\sum_a(\nabla P_{i,j})_a\Omega_{a,b}=x_b(P_{i,\min(b-1,j)}-P_{\max(b+1,i),j})
\end{equation}
and
\begin{equation} \label{pp}
\{ P_{i,j},P_{k,l} \}=\sum_{b=k}^l x_b(P_{i,\min(b-1,j)}-P_{b+1,j}).
\end{equation}
When $j<k$ we get $\sum_{b=k}^l x_b (P_{i,j}-0) = P_{i,j}P_{k,l}$. When $k\leq j \leq l$ equation (\ref{pp}) becomes
\begin{align*}
\sum_{b=k}^j x_b(P_{i,b-1}-P_{b+1,j}) + \sum_{b=j+1}^l x_b(P_{i,j}-0)
&=\sum_{b=k}^j x_b(P_{i,k-1}+P_{k,b-1}-P_{b+1,j})+P_{i,j}P_{j+1,l}\\
&=P_{i,k-1}P_{k,j}+P_{i,j}P_{j+1,l},
\end{align*}
as $\sum_{b=k}^j x_b(P_{k,b-1}-P_{b+1,j})=\{P_{k,l},P_{k,l}\}=0$. Similarly, when $j>l$ equation (\ref{pp}) becomes
\[
\sum_{b=k}^l x_b(P_{i,b-1}-P_{b+1,j})=P_{k,l}(P_{i,k-1}-P_{l+1,j})+\{P_{k,l},P_{k,l}\}.
\]
\hfill $\square$
\item[(\ref{PQ})] Using (\ref{PO}) and
\begin{equation} \label{nQb}
(\nabla Q_{k,l})_b=\begin{cases}0 & b<k \\ (-1)^{b+k+1}\dfrac{Q_{k,l}}{x_b} & k\leq b \leq l \\ 0 & b > l, \end{cases}
\end{equation}
we get, by changing the order of summation,
\begin{align*}
\{ P_{i,j},Q_{k,l} \}&=\left(\sum_{b=k}^l (P_{i,\min(b-1,j)}-P_{\max(b+i,i),j})(-1)^{b+k+1}\right)Q_{k,l}\\
&=\left(\sum_{a=i}^j\left(\sum_{b=k}^{\min(a-1,l)}(-1)^{k+b}+\sum_{b=\max(k,a+1)}^l\right) x_a\right)Q_{k,l}\\
&=\left(\sum_{a=i}^j c^{k,l}_a x_a \right) Q_{k,l}.
\end{align*}
\hfill $\square$
\item[(\ref{QQ})] Taking the sum $\sum_{a=1}^n$ of the product of $(\nabla Q_{i,j})_a$ and $\sum_b \Omega_{a,b}(\nabla Q_{k,l})_b = \{ x_a,Q_{k,l} \}=c_a^{k,l} x_a Q_{k,l}$ yields the result.
\hfill $\square$
\end{enumerate}

\section{More general contraction}
Consider the $n=6$ dimensional LV-system of the form \eqref{LVgen}, with
\[
\r=\begin{pmatrix} u\\v\\v\\v\\v\\v\end{pmatrix},\quad 
\A=\begin{pmatrix}
0 & c  & c  & c  & c  & c  
\\
 -c  & 0 & d  & d  & e  & -e  
\\
 -c  & -d  & 0 & d  & e  & -e  
\\
 -c  & -d  & -d  & 0 & e  & -e  
\\
 -c  & -e  & -e  & -e  & 0 & -e  
\\
 -c  & e  & e  & e  & e  & 0 
\end{pmatrix}.
\]
By introducing variables $\b{y}=(x_1,x_2+x_3+x_4,x_5,x_6)$ the system contracts to $\dot{y}_i = y_i(\tilde{r}_i +  \sum_{j=1}^{4} \tilde{a}_{ij} x_j)$, with
\[
\tilde{\r}=\begin{pmatrix} u\\v\\v\\v\end{pmatrix},\quad  
\tilde{\A}=\begin{pmatrix}
0 & c  & c  & c  
\\
 -c  & 0 & e  & -e  
\\
 -c  & -e  & 0 & -e  
\\
 -c  & e  & e  & 0 
\end{pmatrix}.
\]
This system has Hamiltonian \[
\tilde{H}=-\frac{v}{c}\ln\left(y_1\right)
+\frac{u}{c}\ln\left(\frac{y_3y_4}{y_2}\right)
+\sum_{i=1}^4 y_i
\]
and admits the additional integral \[
\tilde{K}=\frac{y_{2} \left(y_{2}+y_{3}+y_{4}\right)}{y_{3} y_{4}}.
\]
The lifted integrals
\[
H=-\frac{v}{c}\ln\left(x_1\right)
+\frac{u}{c}\ln\left(\frac{x_5x_6}{x_2+x_3+x_4}\right)
+\sum_{i=1}^6 x_i,\quad K=\frac{\left(x_2+x_3+x_4\right) \left(x_2+x_3+x_4+x_{5}+x_{6}\right)}{x_{5} x_{6}}
\]
Poisson commute with respect to the bracket induced by $\A$, and they also Poisson commute with the local integral
\[
L=\frac{\left(x_{2}+x_{3}+x_{4}\right) x_{3}}{x_{2} x_{4}}.
\]
Thus, both systems are Liouville integrable. (The Hamiltonian of the $x$-system is $H+\frac{u}{c}\ln\left(L\right)$.) 
\section{Solutions of 2- and 3-dimensional Lotka-Volterra systems}
We show that the expressions provided in Propositions \ref{p1} and \ref{p2} are solutions to Lotka-Volterra systems \eqref{lv2} and \eqref{lv3}.

\begin{proof}[Proof of Proposition \ref{p1}] 
Solving the first equation of \eqref{lv2} for $D$ yields $D=B'/B-b$. Substituting this into the second equation gives the second order differential equation
\begin{equation} \label{de2}
BB''-(B')^2+B(B-d)(B'-bB)=0.
\end{equation}
Differentiating the expression $F(B(t))=t+c_2$, using the fundamental theorem of calculus, and solving for $B'(t)$ yields
\begin{equation} \label{Bp}
B'=bB\left(W(R(B))+1\right),
\end{equation}
from which we obtain the stated expression for $D(t)$. We differentiate \eqref{Bp}, using the facts that $W'(z)=W(z)(z+zW(z))^{-1}$ and $R'(B)=-(B-d)B'/(bB)$, to get
\[
B''=bB'(W(R(B))+1)-(B-d)B'\frac{W(R(B))}{W(R(B))+1}
\]
which together with \eqref{Bp} implies that our expression for $B$ satisfies \eqref{de2} (and hence \eqref{lv2}).
\end{proof}

\begin{proof}[Proof of Proposition \ref{p2}]
Solving the first equation of \eqref{lv3} for $C+D$ yields $C+D=B'/B-b$. Substituting this into the sum of the other equations gives us an equation we can solve for $C(t)$,
\begin{equation} \label{C}
C=\frac{BB''-(B')^2}{bB^2}-\frac{(B-d)(bB-B')}{bB}.
\end{equation}
and hence for $D(t)$ we find
\begin{equation} \label{D}
D=-\frac{BB''-(B')^2}{bB^2}+\frac{(B-b-d)(bB-B')}{bB}.
\end{equation}
Using these expressions, the system \eqref{lv3} reduces to a single third order nonlinear ODE, which is rather complicated,
\begin{equation} \label{de3}
\begin{split}
&bB\left(B^{2} B''' -3 B B' B'' +2 (B')^{3}\right)+b B^{3} \left(b B - B' \right) \left((B-d)^{2}-B' \right)
\\
&\hspace{2cm}+\left(B'' B -B'^{2}-\left(B -d \right) \left(B b -B' \right) B \right)
\left(\left(B^{2}-d B -B' \right) \left(B b +B' \right)+B'' B\right)=0. 
\end{split}\end{equation}
By differentiating the expression $G(B(t))=t+c_3$ we find
\begin{equation} \label{Bp2}
B'=bB\left(W(T(B))+1\right).
\end{equation}
And, by differentiation of \eqref{Bp2}, and using the same equation to eliminate the explicit dependence on $W(T(B))$, we get
\begin{equation} \label{Bpp}
B''=\frac{(B')^2}{B}+\left(B-d+\frac{bc_2B}{c_1-c_2B}\right)(bB-B').
\end{equation}
By substituting the expression \eqref{Bpp} into \eqref{C} and \eqref{D} one obtains the required expressions for $C$ and $D$. We have, using Maple \cite{Map}, also differentiated \eqref{Bpp} to find an expression for $B'''$. By substitution of this expression, together with \eqref{Bpp}, we find that \eqref{de3} is identically satisfied.
\end{proof}

\end{document}